\documentclass[submission,copyright,creativecommons]{eptcs}
\providecommand{\event}{AFL 2023} 

\usepackage{iftex}

\ifpdf
  \usepackage{underscore}         
  \usepackage[T1]{fontenc}        
\else
  \usepackage{breakurl}           
\fi

\title{A General Approach to Proving Properties of Fibonacci Representations via Automata Theory}
\author{
Jeffrey Shallit\footnote{Research funded by a grant from NSERC, 2018-04118.}
\ and Sonja Linghui Shan
\institute{School of Computer Science,
University of Waterloo,
Waterloo, ON  N2L 3G1,
Canada}
\email{shallit@uwaterloo.ca,
slshan@uwaterloo.ca}
}
\def\titlerunning{Fibonacci Representations via Automata}
\def\authorrunning{J. Shallit and S. L. Shan}

\usepackage[usenames]{color}
\usepackage{amssymb}
\usepackage{amsmath}
\usepackage{amsthm}
\usepackage{amsfonts}
\usepackage{amscd}
\usepackage{graphicx}
\usepackage{soul}
\usepackage{lscape}

\definecolor{webgreen}{rgb}{0,.5,0}
\definecolor{webbrown}{rgb}{.6,0,0}

\usepackage{color}
\usepackage{float}

\usepackage{graphics}
\usepackage{latexsym}
\usepackage{epsf}
\usepackage{breakurl}

\newcommand{\seqnum}[1]{\href{https://oeis.org/#1}{\rm \underline{#1}}}

\def\Enn{\mathbb{N}}

\def\Zee{\mathbb{Z}}

\DeclareMathOperator{\equal}{equal}

\begin{document}
\maketitle

\theoremstyle{plain}
\newtheorem{theorem}{Theorem}
\newtheorem{corollary}[theorem]{Corollary}
\newtheorem{lemma}[theorem]{Lemma}
\newtheorem{proposition}[theorem]{Proposition}

\theoremstyle{definition}
\newtheorem{definition}[theorem]{Definition}
\newtheorem{example}[theorem]{Example}
\newtheorem{conjecture}[theorem]{Conjecture}

\theoremstyle{remark}
\newtheorem{remark}[theorem]{Remark}

\maketitle

\begin{abstract}
We provide a method, based on automata theory, to mechanically prove 
the correctness of many numeration systems based on Fibonacci numbers.   With it, long
case-based and induction-based proofs of correctness
can be
replaced by simply constructing a regular expression (or finite automaton) specifying the rules for valid representations, followed by a short computation.
Examples of the systems that can be handled using our technique include Brown's lazy representation (1965), 
the far-difference representation developed by Alpert (2009), and three representations proposed by Hajnal (2023).
We also provide three additional systems and prove their validity.
\end{abstract}

\section{Introduction}
Given an increasing sequence $(s_n)_{n \geq 0}$ of positive integers,
a numeration system is a way of expressing natural numbers as a linear
combination of the $s_n$.   
Many different numeration systems, 
such as representation in base $k$, or 
the more exotic systems based on the Fibonacci numbers,
have been proposed.  For example,
recall that the Fibonacci numbers, sequence \seqnum{A000045} in the
{\it On-Line Encyclopedia of Integer Sequences} (OEIS),
are defined by the recurrence
$F_n = F_{n-1} + F_{n-2}$ for $n \geq 2$ and the initial values
$F_0 = 0$, $F_1 = 1$.
Consider writing a non-negative integer $n$ as a sum of distinct Fibonacci numbers $F_i$ for $i \geq 2$.  
Some numbers, such as $12$, have only one such representation
$(12 = 8 + 3 + 1 = F_6 + F_4 + F_2)$,
while others have many:
$8 = F_6 = F_5 + F_4 = F_5 + F_3 + F_2$.

There are two very desirable characteristics of a numeration system. 
First, {\it completeness:}  every natural number should have a representation.  Second, {\it unambiguity:}  no natural number should have two or more different representations.   These two goals are typically achieved by restricting the types of representations that are considered valid within the system.  If a system achieves both goals, we say it is {\it perfect}. For Fibonacci representations, various perfect systems have been proposed.

Among all possible perfect systems based on Fibonacci
numbers, one is particularly useful:
the {\it Zeckendorf\/} or {\it greedy\/} representation.
This representation can be computed as follows:   first, choose the largest index $i$ such that $F_i \leq n$.
Then the representation for $n$ is $F_i$ plus the (recursively-computed) representation for $n - F_i$.   The representation for $0$ is the empty sum of $0$ Fibonacci numbers. A simple induction now shows that the greedy algorithm produces a   representation for every natural number, which is evidently unique.

This representation was originally noted by Zeckendorf, but was first published by Lekkerkerker
\cite{Lekkerkerker:1952} and only later by Zeckendorf himself \cite{Zeckendorf:1972}.   It was also anticipated, in much more general form, by
Ostrowski \cite{Ostrowski:1922}.

An alternative (but equivalent) definition of Zeckendorf 
representation is to impose a condition
that valid representations must obey.  For example, we could require that a representation be valid if and only if no two consecutive Fibonacci numbers appear in the sum.  

It is convenient to express arbitrary sums of distinct Fibonacci numbers as strings of digits over a finite alphabet (in analogy with base-$k$ representation).
Let $x = a_1 \cdots a_t$ be a string (or word) made up of integer digits.  We define
its value as a Fibonacci representation as follows:
\begin{equation}
    [x]_F := \sum_{1 \leq i \leq t} a_i F_{t+2-i}.  
\label{eq:fibsum}
\end{equation}
Note that these strings are in ``most-significant-digit'' first format.
For example, $[2101]_F = 2F_5 + F_4 + F_2 = 14$.

It is also useful to define a (partial) inverse
to $[x]_F$.  
By $(n)_F$ we mean the binary string $x$ such that $x$ is the Zeckendorf representation of $n$; alternatively,
such that $[x]_F = n$ and $x$ contains no occurrence of the block $11$.   In what follows, we adopt
this string-based point of view almost exclusively.   We can think of the condition ``no occurrence of the block $11$''
as a {\it rule}, specifying which representations are valid, adopted precisely to guarantee both completeness and unambiguity.

In formal language theory, a language $L$ is a (finite or infinite) collection of strings.   A rule is then
encoded by the language or set of strings that obey the rule.   Completeness
then becomes the assertion that for all $n$ there exists a string $x \in L$ such that $[x]_F = n$, while  
unambiguity becomes the assertion that there do not exist distinct strings $x, y \in L$ such
that $[x]_F = [y]_F$. \footnote{We adopt the convention that two strings are considered to be the same if they differ only in the number of 
leading zeros.   Thus, for example, $[100]_F = [0100]_F = 3$ are
the same representation.} 

Let us look at another example involving the Fibonacci numbers, one that is 
much less well known: the so-called {\it lazy\/} representation \cite{Brown:1965}.
In this system, representation as a sum of Fibonacci numbers corresponds (via Eq.~\eqref{eq:fibsum}) to a binary string having no occurrence of the
block $00$ (where leading zeros are not even considered).   Once again, this rule provides a numeration system that is both complete
and unambiguous \cite{Brown:1965}.  Table~\ref{tab1}
gives both greedy (Zeckendorf) and lazy representations for the first few natural numbers.
\begin{table}[htb]
\begin{center}
\begin{tabular}{c|cccccccccccc}
$n$ & 0 & 1 & 2 & 3 & 4 & 5 & 6 & 7 & 8 & 9 & 10 & 11 \\
\hline
greedy & $\epsilon$ & 1 & 10 & 100 & 101 & 1000 & 1001 & 1010 & 10000 & 10001 & 10010 & 10100 \\
lazy & $\epsilon$ & 1 & 10 & 11 & 101 & 110 & 111 & 1010 & 1011 & 1101 & 1110 & 1111
\end{tabular}
\end{center}
\caption{Greedy and lazy Fibonacci representations.}
\label{tab1}
\end{table}

The greedy and lazy representations are certainly not the only possible perfect numeration systems based on the Fibonacci numbers.   In fact, there are {\it uncountably many} such systems!   These result from making a choice, for all $n$ having at least two different representations as sums of distinct Fibonacci numbers, about which particular representation is chosen to be valid.  (By a result of Robbins \cite{Robbins:1996}, ``most'' numbers have more than one representation as a sum of distinct Fibonacci numbers.)   

If we demand that the set of valid representations forms a regular language---that is, accepted by a finite automaton; see Section~\ref{automata}---there are still infinitely many different systems (although only countably many).  For example, consider choosing the $t$'th largest possible representation for $n$ in lexicographic order (if there are at least $t$), and otherwise the
lexicographically first.  It will follow from results below that, for each $t \geq 0$, this choice gives a regular language $L_t$
of valid representations.

Some natural questions then arise:    given a language $L$ encoding the ``rule'' a representation must obey (such as no occurrence of the block $11$, or no occurrence of the block $00$), how can we determine if the corresponding set of Fibonacci representations is complete and unambiguous?  And if it is complete, how can we efficiently find
a representation for a given number $n$?    Up to now, each new system proposed required a new proof, often a rather tedious case-based proof by induction.  In this paper we provide a {\it general framework\/} for answering these questions ``automatically", via an algorithm, in the case where the language of valid
representations is regular.

These ideas are capable of generalization.  For example, we can also consider representations for all integers $\Zee$, instead of just the natural numbers $\Enn$.
This can be achieved in two distinct ways:
\begin{itemize}
\itemsep-0.2em
\item By allowing a larger digit set, say,  $\{ -1, 0, 1 \}$;

\item By using the so-called negaFibonacci system, based on the Fibonacci numbers of negative index $F_{-n}$ for $n \geq 1$.  
\end{itemize}
Once again, we would like a choice of valid representations that is complete and unambiguous.  

In this paper we show how to decide these properties, provided that the set of valid representations forms a regular language (which is indeed the case for all the proposed systems in the literature).

Here is an outline of the paper.  In Section~\ref{automata}, we explain the basics of automata theory needed to understand the rest of the paper.  In Section~\ref{rep01}, we discuss how to test completeness and ambiguity for systems using digits $0$ and $1$ only.  
In Section~\ref{repneg} we discuss systems
using digits $-1, 0, 1$ only.  In Section~\ref{allint} we discuss representations for all integers, not just the natural numbers.  In Section~\ref{dicto}
we discuss an entirely new type of Fibonacci representation based on dictionary order.   Finally, in Section~\ref{exper} we describe a few of the new Fibonacci representations we found through exhaustive search of small automata.

\section{The decision procedure and {\tt Walnut}}
\label{automata}

We assume the reader is familiar with the basics of automata theory
as discussed, for example, in
\cite{Hopcroft&Ullman:1979}.

The following particular case of a theorem of B\"uchi \cite{Buchi:1960} 
(as later corrected by Bruy\`ere et al. \cite{Bruyere&Hansel&Michaux&Villemaire:1994} is our principal tool in the paper.
\begin{theorem}
There is a decision procedure that, given a first-order logical formula $F$ involving natural numbers, comparisons, automata, and addition, and no free variables, will decide
the truth or falsity of $F$.   Furthermore, if $F$
has free variables, the procedure constructs a DFA
accepting those values of the free variables (in Fibonacci representation) that make $F$ evaluate to {\tt TRUE}.
\label{one}
\end{theorem}
For more information about the specific case
of the decision procedure for Fibonacci representation, see \cite{Mousavi&Schaeffer&Shallit:2016}.

We should explain how automata can process pairs, triples, and generally $k$-tuples of inputs.    This is done by replacing the input alphabet $\Sigma$ with
the alphabet $\overbrace{\, \Sigma \times \Sigma \times \cdots\times \Sigma\,}^{k \rm\ times}$.  In other words, inputs are $k$-tuples of alphabet symbols.  The $i$'th input then corresponds to the concatenation of the $i$'th components of all the $k$-tuples.   Of course, this means that all $k$ inputs have to have the same length; this is achieved by padding shorter inputs, if necessary, with leading zeros. 

The decision procedure of Theorem~\ref{one} has been implemented in free software called {\tt Walnut},
originally created by Hamoon Mousavi \cite{Mousavi:2016}; also see the book
\cite{Shallit:2022}.
We recall some of the basics of {\tt Walnut} syntax:
\begin{itemize}
\itemsep-0.2em
    \item {\tt eval} evaluates a formula with no free variables and returns {\tt TRUE} or {\tt FALSE};
    {\tt def} defines an automaton for future use; {\tt reg} defines a regular expression.
    \item In a regular expression, the period
    is an abbreviation for the entire alphabet.
    \item {\tt \&} is logical {\tt AND},
    {\tt |} is logical {\tt OR}, 
    {\tt =>} is logical implication,
    {\tt <=>} is logical {\tt IFF},
    {\tt \char'176} denotes logical {\tt NOT}.
    \item {\tt A} denotes $\forall$ (for all);
    {\tt E} denotes $\exists$ (there exists).
    \item {\tt ?msd\_fib} tells {\tt Walnut} to evaluate an arithmetic expression using
    Fibonacci representation.
\end{itemize}

We use {\tt Walnut} to do the computations needed to verify that a given system is complete
and unambiguous.
For much more about {\tt Walnut}, including a link to download it, visit\\
\centerline{\url{https://cs.uwaterloo.ca/~shallit/walnut.html} \ . }

\section{Representation of natural numbers using digits 0 and 1 only}
\label{rep01}

In this section we consider
representations of the natural
numbers by Fibonacci numbers using digits
$0$ and $1$ only.

The first step is to find an automaton that can convert from an arbitrary Fibonacci representation
to the greedy or Zeckendorf representation.   
To do this we use the following simple observation:

\begin{proposition}
We can convert a binary string $x$ to a Zeckendorf representation $y$
for the same number using the following algorithm:
first append a 0 on the front, if necessary.  Then scan the string from left to right, replacing each
occurrence of ``\,$011$" successively with ``\,$100$".
\end{proposition}

\begin{proof}
Clearly each such replacement does not change the value of $[x]_F$.   The algorithm terminates
because each replacement lowers the total number of $1$'s by $1$.   Finally, the algorithm
clearly cannot result in two consecutive $1$'s, because it introduces two consecutive $0$'s, only
the second of which can later change to a $1$.
\end{proof}

We can implement this idea as a DFA  $C$ that takes two inputs in parallel, $x$ and $y$, and accepts
if and only if both $[x]_F = [y]_F$ and $y$ is a valid Zeckendorf representation; that is, it contains no two consecutive $1$'s.    It suffices
to keep track of $[x']_F - [y']_F$ for the prefix $x'$ of $x$ seen so far, and similarly
for the prefix $y'$ of $y$ seen so far.   Note that we assume that $x$ and $y$ have the
same length, with the shorter of the two prefixed by leading zeros, if necessary.   We can think of this as a ``converter'' or ``normalizer''
that allows us to turn arbitrary Fibonacci representations into Zeckendorf representations.  It is depicted
in Figure~\ref{fig: bfsc transition diagram}.
\begin{figure}[h]
\centering
    \includegraphics[scale=.28]{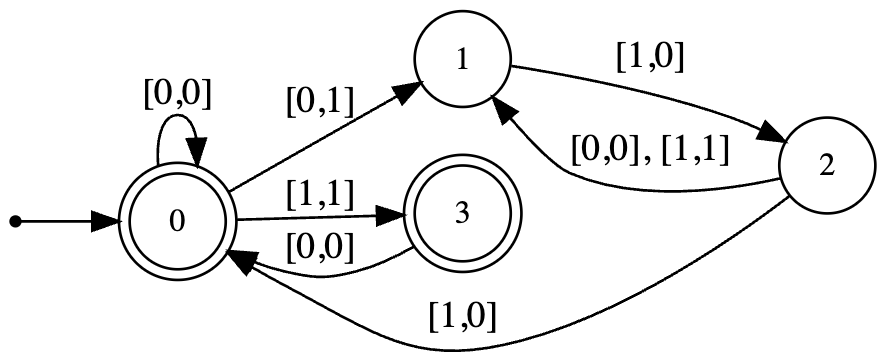}
    \caption{DFA $C$ for conversion to the Zeckendorf representation.}
    \label{fig: bfsc transition diagram}
\end{figure}
This automaton was given by 
Berstel \cite{Berstel:2001} in a slightly different form.  Also see \cite{Shallit:2021c}.

As an example, consider the input
$[0,1][1,0][1,0][1,1][0,0]$ to $C$, whose
first components spell out $x=01110$ and
whose second components spell out $y=10010$.
Starting in state $0$, the automaton visits,
successively, states $1,2,0,3,0$, and hence
accepts---as it should, since $[x]_F = [y]_F$.  

We now state one of our main results.
\begin{theorem}
There is an algorithm that,
    given rules that specify which representations are valid (in the form of a regular language $L$ of all valid representations), will decide if the corresponding numeration system based on the Fibonacci numbers is complete and unambiguous for $\Enn$.
\end{theorem}

\begin{proof}
Using Theorem~\ref{one}, it suffices
to express the properties of completeness and unambiguity as a first-order logic formula $F$.   Once this is done, the decision algorithm can determine if $F$ is true or false.

Completeness says every integer has a representation in $L$.  We can express this as follows:
\begin{equation}
\forall n \ \exists x\ x \in L \ \wedge \ [x]_F = n,
\label{log1}
\end{equation}

Unambiguity says that no integer has two distinct representations in $L$.  We can express this as follows:
\begin{equation}
    \neg\exists x, y \in L \ (\neg \equal(x,y)) \ \wedge \ 
[x]_F = [y]_F.
\label{log2}
\end{equation}
Here $\equal$ means that $x$ and $y$ are the same, up to leading zeros.
\end{proof}

Furthermore, if $L$ is a regular language that provides a system that is complete, we can find a representation in $L$ for $n$ efficiently.
The first step is to represent $n$ in Fibonacci representation, say using the greedy algorithm.  Construct a new automaton from {\tt fcanon} by using two intersections.   The first intersection is with an automaton with a first component that belongs to $L$, while
the second component is arbitrary.   The second intersection is with an automaton where the first component is arbitrary, and the second is of the form 
$0^*(n)_F$.  This gives a new automaton
of $O(\log n)$ states, and it now suffices to find any accepting path (a path from the initial state
to the final state).  This can be done in linear time in the number of states using depth-first or
breadth-first search.   This gives us an $O(\log n)$ algorithm to find a representation.    Thus we have proved:
\begin{theorem}
Suppose $L$ is a regular language.   If $L$ is complete, we can find a representation for an integer $n$
in $O(\log n)$ time.
\end{theorem}

\begin{remark}
Here we use the convention of the so-called ``word RAM" model, where we assume that $n$ fits
in a single machine word, or more generally that we can
perform basic operations on integers with $O(\log n)$ bits in unit time.
\end{remark}

All this can be carried out mechanically with {\tt Walnut}.  Here all we have to do is define the language $L$ of valid representations (say, with a regular expression) and type in the {\tt Walnut} commands corresponding to the two logical assertions \eqref{log1} and \eqref{log2}.
We illustrate this with two examples.

The first is the lazy representation mentioned previously, and discussed first by Brown \cite{Brown:1965}.  The first step is to give a
regular expression defining a valid representation
in Brown's system:
\begin{verbatim}
reg lazyExclude {0,1} "0*1(0|1)*00(0|1)*":
def lazy "~$lazyExclude(s)":
\end{verbatim}
This gives a $4$-state automaton testing the
lazy criterion that is depicted in Figure~\ref{lazyr}.
\begin{figure}[htb]
\centering
    \includegraphics[width=4.0in]{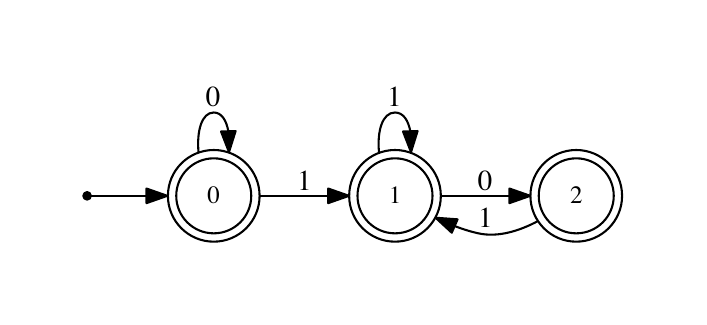}
    \caption{DFA for Brown's lazy representation.}
    \label{lazyr}
\end{figure}

We test the completeness and unambiguity for Brown's system as follows. 
\begin{verbatim}
reg equal {0,1} {0,1} "([0,0]|[1,1])*":
eval brown1 "?msd_fib An Es $fcanon(s,n) & $lazy(s)":
eval brown2 "?msd_fib ~En,s,t $lazy(s) & $lazy(t) & (~$equal(s,t)) 
    & $fcanon(s,n) & $fcanon(t,n) ":
\end{verbatim}
Both return {\tt TRUE}.
Given these results, we have now proven that the lazy representation is complete and unambiguous. 

For a second example, see the Appendix.

\section{Representation using digits $-1$, $0$, and $1$}
\label{repneg}

We now turn to representations using digits $-1$, $0$, and $1$ in the Fibonacci system.   

Recently, Hajnal \cite{Hajnal:2023} described three Fibonacci representations 
using Eq.~\eqref{eq:fibsum} to associate 
a string $x = e_t e_{t-1} \cdots e_2 \in \{ -1,0,1\}^*$ 
with a natural number $n$:  {\it alternating, even, and odd.}   Using
induction and a case-based argument, he proved that each of these three representations is
complete and unambiguous.   

Using automata, we can replace his rather long arguments with our general approach.
We first describe each of his systems, and show that
the set of valid representations for all natural numbers is a regular language.
 
The alternating representation requires 
a representation to fulfill four conditions:
\begin{enumerate}
\itemsep-0.2em 
    \item 
    the most significant nonzero term is positive, 
    \item
    two adjacent nonzero terms cannot be of the same sign, 
    \item 
    two adjacent nonzero terms have at least one zero in between, and
    \item
    if there are two or more nonzero terms, then 
    there has to be at least two zeros 
    between the last and the second last nonzero terms.
\end{enumerate}
We denote a number $n$ in this representation as $ [n]_{A} $.
For example, $ [9]_{A} = 10\Bar{1}001 $, where $ \Bar{1} $ is used for $-1$.

For the alternating representation,
we can use the following {\tt Walnut} code:
\begin{verbatim}
reg altInclude1 {-1,0,1} "(0*|0*1.*)":
reg altExclude1 {-1,0,1} ".*(10*1|[-1]0*[-1]).*":
reg altExclude2 {-1,0,1} ".*(1[-1]|[-1]1).*":
reg altInclude2 {-1,0,1} "(0*|0*10*|.*(100+[-1]|[-1]00+1)0*)":
def alt "~$altExclude1(s) & ~$altExclude2(s) & $altInclude1(s) & $altInclude2(s)":
\end{verbatim}
The result is an automaton of 12 states that checks whether an input over the alphabet $\{-1,0,1\}$ is alternating, and is illustrated
in Figure~\ref{altfig}.
\begin{figure}[htb]
\centering
    \includegraphics[width=5in]{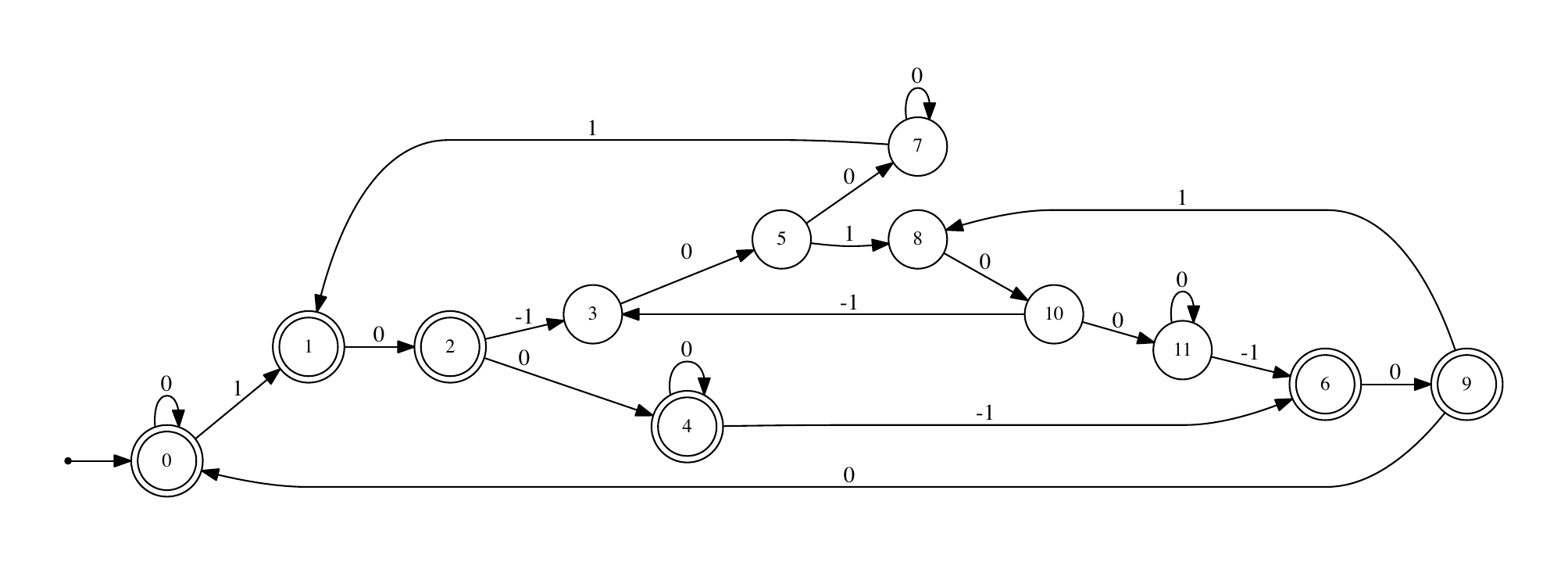}
    \caption{DFA for the alternating condition.}
    \label{altfig}
\end{figure}

The even representation requires three conditions:
\begin{enumerate}
\itemsep-0.2em
    \item 
    the most significant nonzero term is positive, 
    \item
    only positions indexed with even numbers, such as $e_2$, can have nonzero terms, and 
    \item 
    two adjacent nonzero terms cannot both be $-1$.
\end{enumerate}
We denote a number $n$ in this representation as $ [n]_{E} $.
For example, $ [14]_{E} = 10\Bar{1}0001 $.
\begin{verbatim}
reg evenInclude {-1,0,1} "(0*|0*1(0[-1]|01|00)*)":
reg evenExclude {-1,0,1} ".*[-1]0*[-1].*":
def even "$evenInclude(s) & ~$evenExclude(s)":
\end{verbatim}
This gives us a 5-state automaton to check the even condition, which is
illustrated in Figure~\ref{evenfig}.
\begin{figure}[htb]
\centering
    \includegraphics[width=5in]{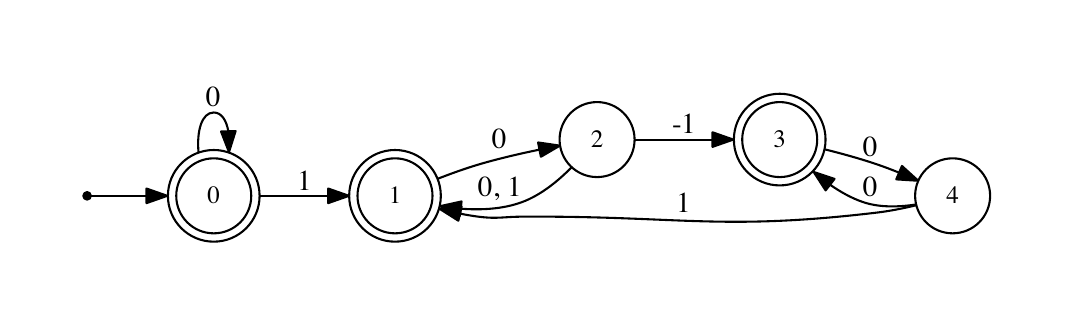}
    \caption{DFA for the even condition.}
    \label{evenfig}
\end{figure}

The odd representation adds an epsilon term to the sum in Eq.~\eqref{eq:fibsum},
therefore associating 
a string $ e_t e_{t-1} \cdots e_2 \epsilon $,
where $ \epsilon \in \{-1,0\} $, with a number $n$.
The odd representation requires the string to meet three conditions:
\begin{enumerate}
\itemsep-0.2em
    \item 
    the most significant nonzero term is positive, 
    \item
    only positions indexed with odd numbers (such as $e_3$) and the epsilon term are allowed to be nonzero, and 
    \item 
    two adjacent nonzero terms cannot both be $-1$.
\end{enumerate}
We denote a number $n$ in this representation as $ [n]_{O} $.
For example, $ [14]_{O} = 100010\Bar{1} $,
where $ \Bar{1} $ is used for $ \epsilon = -1 $.

We express the odd representation conditions in \texttt{Walnut} as follows. 
Notice we relax the third condition (required in \cite{Hajnal:2023}) slightly by limiting its application to only the string $ e_t e_{t-1} \cdots e_2 $ without the $\epsilon$ term.
\begin{verbatim}
reg oddInclude {-1,0,1} "(0*|0*10([-1]0|10|00)*)":
reg oddExclude {-1,0,1} ".*[-1]0*[-1].*":
def odd "$oddInclude(s) & ~$oddExclude(s)":
\end{verbatim}

This gives us a 5-state automaton to check the odd condition, which is
illustrated in Figure~\ref{oddfig}.
\begin{figure}[htb]
\centering
    \includegraphics[width=5in]{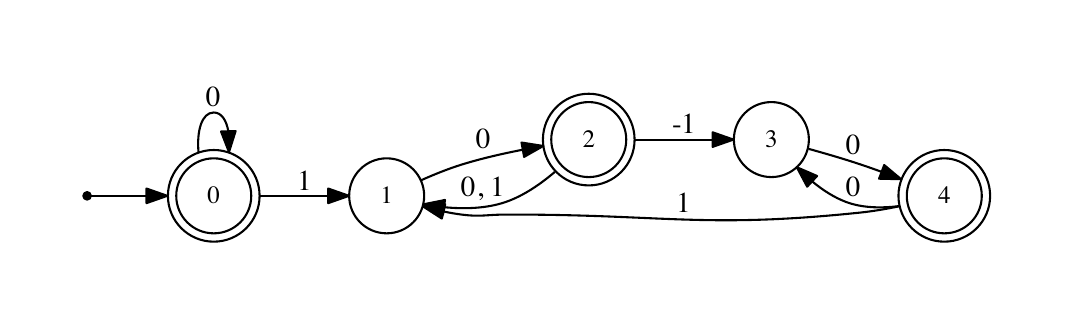}
    \caption{DFA for the odd condition.}
    \label{oddfig}
\end{figure}

It now remains to use our technique to show that these representations are all complete
and unambiguous.   In order to do this, we need a ``converter'' automaton that can
compare representations using digits $-1,0,1$ to ordinary Zeckendorf representation.
We can construct such an automaton based on {\tt fcanon} as follows.   The idea is to use one automaton to ``select'' the positive digits of a representation, another one to ``select'' the negative digits, and then do an (implicit) subtraction
to obtain the value of the representation.
\begin{verbatim}
reg posdigits {-1,0,1} {0,1} "([1,1]|[-1,0]|[0,0])*":
reg negdigits {-1,0,1} {0,1} "([-1,1]|[1,0]|[0,0])*":
def fcanon2 "?msd_fib Et,u,w,s $negdigits(x,t) & $posdigits(x,u) &
   $fcanon(t,w) & $fcanon(u,s) & z+w=s":
\end{verbatim}
This gives a 24-state automaton
{\tt fcanon2}, the analogue of {\tt fcanon}, for doing the conversion.

Let us now check that the alternating representation of Hajnal is both complete and unambiguous.
\begin{verbatim}
reg same {-1,0,1} {-1,0,1} "([-1,-1]|[0,0]|[1,1])*":
eval altRep1 "?msd_fib An Es $fcanon2(s,n) & $alt(s)":
# evaluates to TRUE, 4 ms
eval altRep2 "?msd_fib ~En,s,t $alt(s) & $alt(t) & (~$same(s,t))
    & $fcanon2(s,n) & $fcanon2(t,n)":
# evaluates to TRUE, 31 ms
\end{verbatim}

Similarly, we can check the even and odd representations, as follows:
\begin{verbatim}
eval evenRep1 "?msd_fib An Es $fcanon2(s,n) & $even(s)":
# evaluates to TRUE, 1 ms
eval evenRep2 "?msd_fib ~En,s,t $even(s) & $even(t) & (~$same(s,t))
    & $fcanon2(s,n) & $fcanon2(t,n)":
# evaluates to TRUE, 4 ms
eval oddRep1 "?msd_fib An (Es $fcanon2(s,n) & $odd(s)) | 
   (Et $fcanon2(t,n+1) & $odd(t))":
# evaluates to TRUE, 7 ms
eval oddRep2 "~En,s,t $odd(s) & $odd(t) & (~$same(s,t))
    & $fcanon2(s,n) & $fcanon2(t,n)":
# evaluates to TRUE, 4 ms
\end{verbatim}
This completes our proof that all three systems of Hajnal are complete and unambiguous.

\begin{remark}
We noticed, by testing the following, that this representation is also complete 
if $ \epsilon \in \{1,0\} $ instead of $ \epsilon \in \{-1,0\} $ as required in \cite{Hajnal:2023}.
\begin{verbatim}
eval oddRep3 "?msd_fib An 
    (Es $fcanon2(s,n) & $odd(s)) | (Et $fcanon2(t,n-1) & $odd(t))":
# evaluates to TRUE, 4 ms
\end{verbatim}
\end{remark}

\section{Representations for all integers} 
\label{allint}

In this section we investigate two different
ways to represent {\it all\/} integers (not just
the natural numbers) using Fibonacci representations.

Alpert \cite{Alpert:2009} described a {\it far-difference representation} for Fibonacci numbers
that writes {\it every integer\/} (not just
the natural numbers), with a Fibonacci
numeration system using the digits
$-1,0,1$.     
In Alpert's system, the far-difference representation requires the string to have
\begin{enumerate}
\itemsep-0.2em
    \item 
    at least three zeros between any two nonzero terms of the same sign, and
    \item 
    at least two zeros between any two nonzero terms of different signs.
\end{enumerate}
We use $ [n]_A $ to denote a natural number in this representation:
for example, $ [-38]_A = \Bar{1}000\Bar{1}001 $. 
One nice feature of Alpert's system is that it
is very easy to negate an integer:  all we have to
do is change the sign of each digit.\footnote{The three systems proposed by Hajnal also exhibit this property. Therefore, if we exclude the condition stating "the most significant nonzero term is positive" from the three systems, they can be perfect representations for all integers.}

We express the far-difference representation conditions in \texttt{Walnut} as follows. 
\begin{verbatim}
reg exclude1 {-1, 0, 1} ".*([-1][-1]|[-1]0[-1]|[-1]00[-1]|11|101|1001).*":
reg exclude2 {-1, 0, 1} ".*([-1]1|1[-1]|10[-1]|[-1]01).*":
def alpert "~$exclude1(s) & ~$exclude2(s)":
\end{verbatim}
This gives a $7$-state automaton that checks
the Alpert condition, as illustrated in
Figure~\ref{alp}.
\begin{figure}[h]
\centering
    \includegraphics[width=5in]{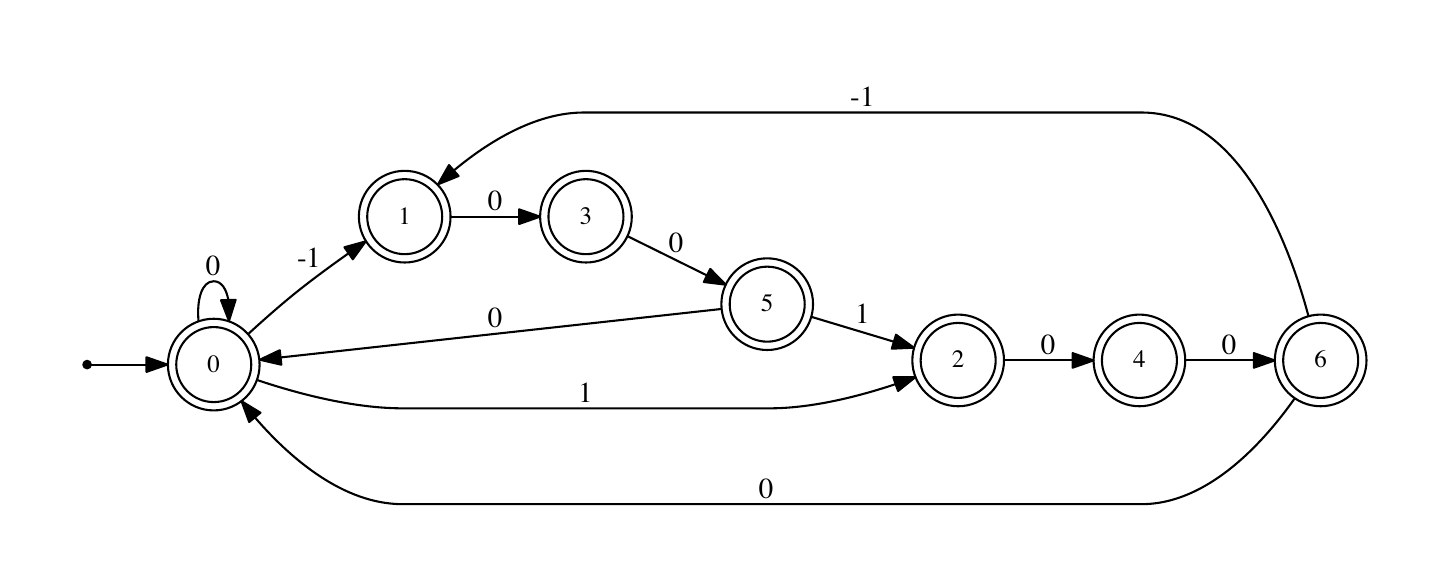}
    \caption{DFA for the Alpert conditions.}
    \label{alp}
\end{figure}

To check completeness and ambiguity, we have to check
positive and negative integers separately.
In addition to {\tt fcanon2}, we need an
automaton {\tt fcanon2\_neg} that takes a
string $x$ over the alphabet $\{-1,0,1\}$
and a natural number $n \geq 0$ as input
and accepts if $[x]_F = -n$.
\begin{verbatim}
def fcanon2_neg "?msd_fib Et,u,w,s $negdigits(x,t) & $posdigits(x,u) &
   $fcanon(t,w) & $fcanon(u,s) & z+s=w":
\end{verbatim}

We can then prove the completeness and unambiguity of this system as follows. 
\begin{verbatim}
eval farDiff1_pos "?msd_fib An Es $fcanon2(s,n) & $alpert(s)":
eval farDiff1_neg "?msd_fib An Es $fcanon2_neg(s,n) & $alpert(s)":
# both evaluate to TRUE, 3 ms
eval farDiff2_pos "?msd_fib ~En,s,t $alpert(s) & $alpert(t)
    & (~$same(s,t)) & $fcanon2(s,n) & $fcanon2(t,n)":
eval farDiff2_neg "?msd_fib ~En,s,t $alpert(s) & $alpert(t)
    & (~$same(s,t)) & $fcanon2_neg(s,n) & $fcanon2_neg(t,n)":
# both evaluate to TRUE, 9 ms
\end{verbatim}
Thus we have easily verified the correctness
of Alpert's conditions.   

Bunder \cite{Bunder:1992} invented a different numeration system for all integers, called the negaFibonacci system.  In this system, we write integers as
a sum of distinct Fibonacci numbers with negative
indices, subject to the condition that no two consecutive Fibonacci numbers can be used.   Since $F_{-n} = (-1)^{n+1} F_n$
for $n \geq 1$, this is the same as enforcing
the requirement in a
Fibonacci representation 
$a_t F_t + \cdots + a_2 F_2 + a_1F_1$ with
digits $a_i \in \{-1,0, 1\}$, (a) only the
terms with odd indices are allowed to be
positive and only the terms with even indices are
allowed to be negative and (b) no two consecutive
nonzero digits can appear.
We can enforce this condition as follows:
\begin{verbatim}
reg bunder1 {-1,0,1} ".*1.(..)*":
reg bunder2 {-1,0,1} ".*[-1](..)*":
reg bunder3 {-1,0,1} ".*((1[-1])|([-1]1)).*":
def bunder "~$bunder1(x) & ~$bunder2(x) & ~$bunder3(x)":
\end{verbatim}
which gives the automaton in 
Figure~\ref{bund}.
\begin{figure}[h]
\centering
    \includegraphics[width=5in]{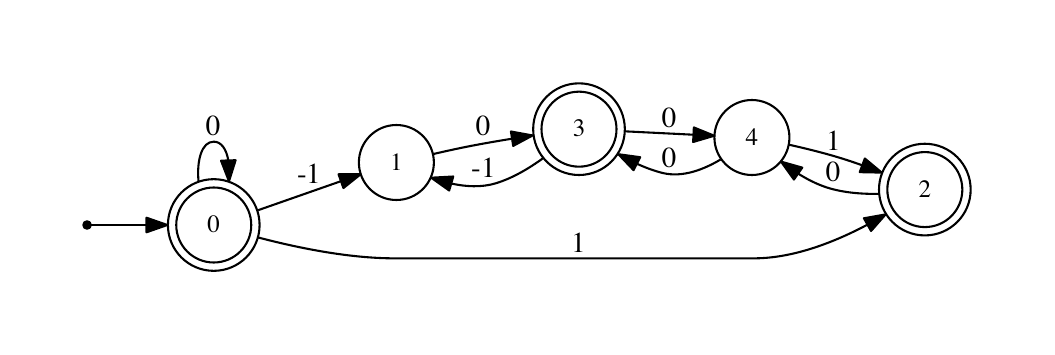}
    \caption{DFA for the Bunder conditions.}
    \label{bund}
\end{figure}
We can then check completeness and unambiguity much as we did for Alpert's system, but there is a new
wrinkle:  representations have an extra digit
at the end, corresponding to the term $a_1 F_1$,
that must be taken care of.  To do this
we introduce a ``shifter" automaton that shifts a representation to the right, and a ``lastbit"
that determines if the last bit of a 
representation is $1$ or $0$.
The shifter is called {\tt rshiftfib} and is
displayed in Figure~\ref{rshift}.
\begin{figure}[htb]
\centering
    \includegraphics[width=4in]{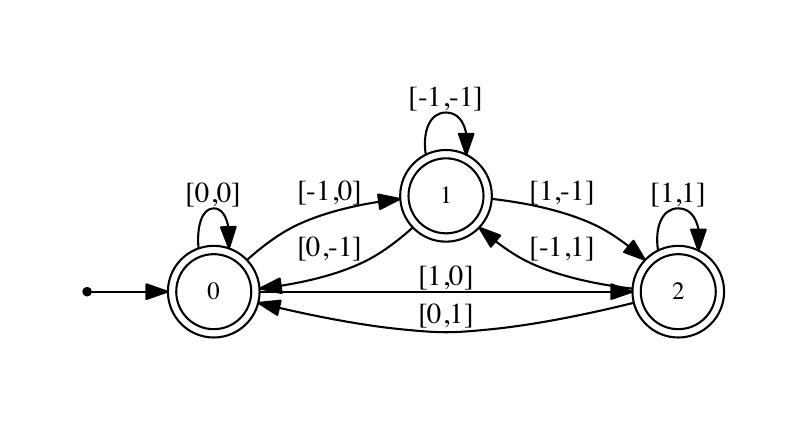}
    \caption{Shifter automaton.}
    \label{rshift}
\end{figure}

Then Bunder's representation can be verified to be complete and unambiguous, as follows:
\begin{verbatim}
reg lastbit {-1,0,1} {0,1} "([0,0]|[1,0]|[-1,0])*([1,1]|[0,0])":
def fcanon3 "?msd_fib Et,u,m $rshiftfib(x,t) &
   $lastbit(x,u) & $fcanon2(t,m) & z=m+u":
def fcanon3_neg "?msd_fib Et,u,m $rshiftfib(x,t) &
   $lastbit(x,u) & $fcanon2_neg(t,m) & z=m-u":
eval bunder1_pos "?msd_fib An Es $fcanon3(s,n) & $bunder(s)":
eval bunder1_neg "?msd_fib An Es $fcanon3_neg(s,n) & $bunder(s)":
# both evaluate to TRUE, 1 ms
eval bunder2_pos "?msd_fib ~En,s,t $bunder(s) & $bunder(t)
    & (~$same(s,t)) & $fcanon3(s,n) & $fcanon3(t,n)":
eval bunder2_neg "?msd_fib ~En,s,t $bunder(s) & $bunder(t)
    & (~$same(s,t)) & $fcanon3_neg(s,n) & $fcanon3_neg(t,n)":
# both evaluate to TRUE, 12 ms
\end{verbatim}

\section{Maximum dictionary order representation}
\label{dicto}

In this section we consider an entirely new Fibonacci representation based on dictionary order. 
We first introduce how strings are compared in dictionary order.
Let $s = s_1s_2 \cdots s_m$ and $t = t_1t_2 \cdots t_n$ where $m \leq n$ be two strings. 
Let $i$ such that $1 \leq i \leq m$ be the first position where $s_i \neq t_i$.
If $s_i < t_i$, then $s < t$ in dictionary order; otherwise $s > t$.
For example, $1\underline{0}11 < 1\underline{1}00$, but $10\underline{1}1 > 10\underline{0}1$.
If there is no such position $i$, then either $s = t$ or $s$ is a proper prefix of $t$.
In this latter case we say $s<t$.
For example, $110 = 110$ and $110 < 1100$.

Consider a representation of natural numbers by always choosing the {\it largest string representation in dictionary order\/} for every number.   
Since every number has a Fibonacci-based representation, 
the representation is complete. 
Since we choose only one Fibonacci-based representation for each number, the representation is unambiguous.   Representations of the first few numbers are given in Table~\ref{dictable}.
\begin{table}[htb]
\begin{center}
    \begin{tabular}{c|ccccccccccc}
    $n$ & 1 & 2 & 3 & 4 & 5 & 6 & 7 & 8 & 9 & 10 & 11 \\
    \hline
    $(n)_D$ & 1 & 10 & 11 & 101 & 110 & 111 & 1010 & 1100 & 1101 & 1110 & 1111 
    \end{tabular}
\end{center}
\caption{Representations for the first few numbers.}
\label{dictable}
\end{table}

We now show that
\begin{theorem}
The set of largest Fibonacci representations in dictionary order forms a regular language.
\end{theorem}

\begin{proof}
The idea is to construct a comparator DFA $C_D$ that
can take two representations in parallel and decide if one is greater than the other, in dictionary order.   

In order to take two representations in parallel, they would have to be the same length, and therefore the shorter one would have to be padded with leading zeros to make it the same length as the longer one.  In this case, it is not hard to see that no automaton can do the needed comparison.

However, in our case, we can take advantage of the following fact:  two Fibonacci representations for the same number cannot be of wildly different lengths.
\begin{lemma}
\label{lem: diffby1}
    The lengths of two Fibonacci-based representation strings for the same natural number differ by one at most (not counting leading zeros).
\end{lemma}
\begin{proof}
    Let $s$ and $t$ be two Fibonacci representations for a natural number $m$.
    Without loss of generality, assume that $s$ is longer.
    Suppose the leading $1$ digit of $s$ corresponds to $F_i$.
    If $s$ and $t$ differ in length by more than one,
    then $t$ is a sum of some $F_j$'s where $j \leq i-2$.
    Now a classic identity on Fibonacci numbers states that $\sum_{0 \leq j \leq n} F_j = F_{n+2}-1$.
    Using this relation, we conclude that 
    $\sum_{j=2}^{i-2} F_j = F_i - 2 < F_i$.
    Therefore $s$ and $t$ do not represent the same number.
\end{proof}

Using this fact, it is indeed possible to compare two strings in dictionary order with an automaton.
\begin{figure}[h]
\centering
    \includegraphics[width=4in]{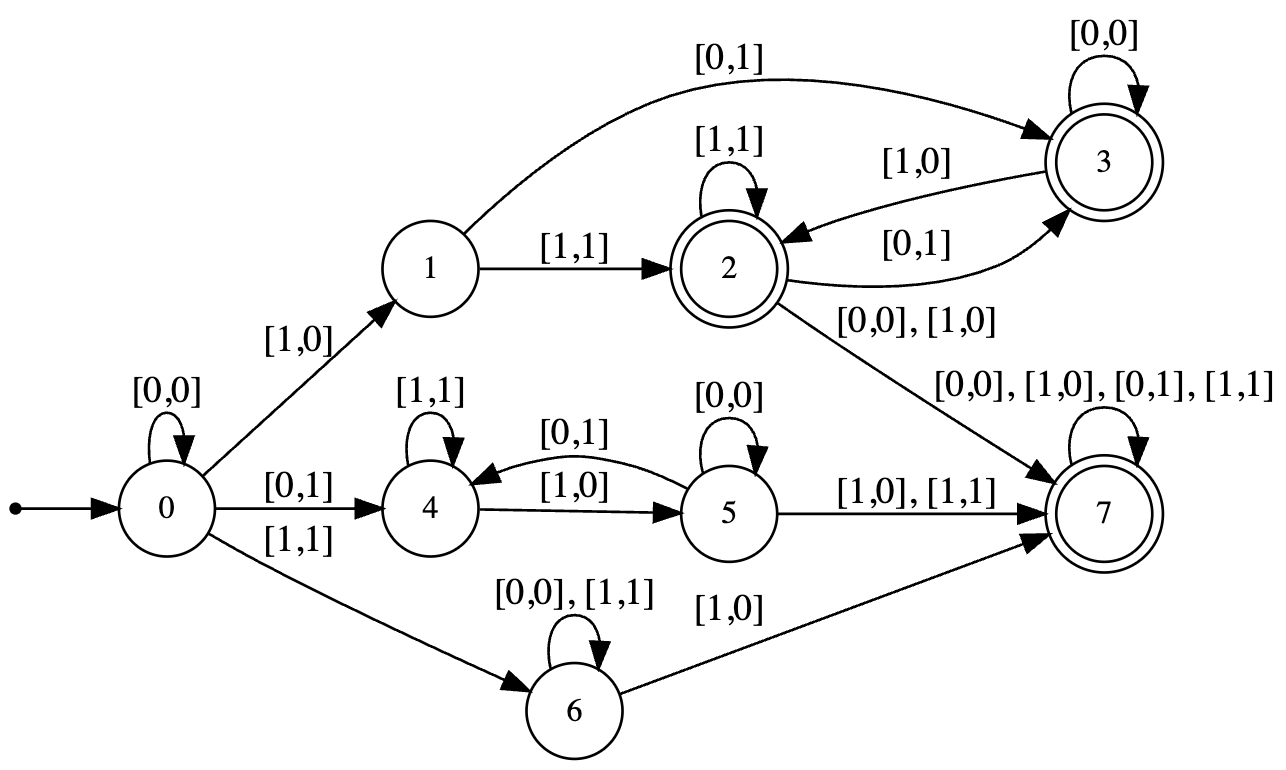}
    \caption{DFA $C_D$ for comparing strings in dictionary order.}
    \label{fig: dict1stGreater transition diagram}
\end{figure}
It is shown in Fig.~\ref{fig: dict1stGreater transition diagram} and
takes two inputs in parallel, $s'$ and $t'$.
Let $s$ and $t$ be $s'$ and $t'$ without leading zeros. 
The DFA $C_D$ accepts if and only if $s$ is greater than $t$ in dictionary order.
We have three cases to consider: $|s| > |t|$, $|s| < |t|$, and $|s| = |t|$.  
We now discuss how the $8$ states of $C_D$ relate to these $3$ cases.
\begin{itemize}
\itemsep-0.2em
    \item 
    State $0$ is the initial state. 
     
    \item
    State $1$ is reached if $|s| > |t|$; that is, if $s'$ starts with $1$ and $t'$ starts with $01$.
    \item
    State $2$ is reached when
    $|s| > |t|$, $s$ ends in $1$, and based on the inputs so far, 
    $t$ is a proper prefix of $s$ therefore $s > t$.
    \item
    State $3$ is reached when
    $|s| > |t|$, $s$ ends in $0$, and based on the inputs so far, 
    $t$ is a proper prefix of $s$ therefore $s > t$.
    \item 
    State $4$ is reached when
    $|s| < |t|$ and $t$ ends in $1$, and based on the inputs so far, 
    $s$ is a proper prefix of $t$ therefore $s < t$.
    \item 
    State $5$ is reached when
    $|s| < |t|$ and $t$ ends in $0$, and based on the inputs so far, 
    $s$ is a proper prefix of $t$ therefore $s < t$.
    \item 
    State $6$ is reached when
    $|s| = |t|$ and, based on the inputs so far, 
    we have $s = t$.
    \item 
    State $7$ is one of the accepting states.
    It is reached when we can identify a position $i$ such that $s_i > t_i$ .
    Additional symbols read, starting from this
    state, cannot
    change the comparison result.
\end{itemize}
It is now easy to verify that the transitions
maintain the invariants corresponding to each
state, and we leave this to the reader.
\end{proof}

Using the comparator automaton, we can build a DFA $D$ that finds the maximum dictionary order representation for each natural number.
The automaton $D$ takes two inputs in parallel: a number $n$ in Zeckendorf representation and a string $s \in \{0,1\}^*$;
and it only accepts if, out of all Fibonacci-based representations of $n$, the string $s$ is the greatest based on dictionary order.
We implement $D$ in \texttt{Walnut} as follows.
\begin{verbatim}
def dictOrder "$fcanon(s,n) & (At $fcanon(t,n) => ($dGreater(s,t)|$equal(s,t)))":
\end{verbatim}
Here \texttt{dictOrder} implements the automaton $D$;
\texttt{fcanon}, the automaton $C$; and
\texttt{dGreater}, the automaton $C_D$.
The resulting automaton has $7$ states and is depicted in Figure~\ref{dictaut}.
\begin{figure}[htb]
\centering
    \includegraphics[width=4in]{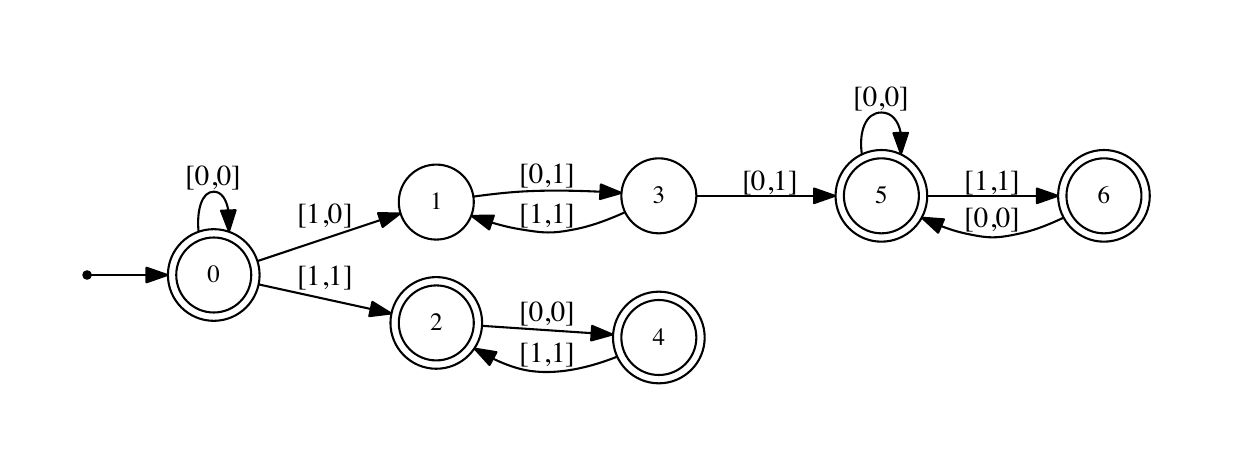}
    \caption{DFA $D$ for converting to dictionary order representation.}
    \label{dictaut}
\end{figure}

\section{Finding new perfect systems of small complexity via exhaustive search}
\label{exper}

We see that a Fibonacci-based representation of natural numbers 
can be represented by a language over the binary alphabet $\{0,1\}$. 
If the language is regular, 
we can express it with a DFA and test its completeness and unambiguity in \texttt{Walnut}.
For example, the Zeckendorf representation can be expressed as a $3$-state DFA and 
the Brown one, a $4$-state DFA. 
Therefore we were curious about whether there exist other DFAs with a small number of states
that can qualify as complete and unambiguous representations. 
We conducted an exhaustive search to find such automata
and found a surprising number of them. 
If we allow up to $7$ states, we found more than $28$ new complete and unambiguous representations.\footnote{There could be more as the heuristics we used to trim our search tree can sometimes exclude eligible representations if, for two numbers $m,n$ where $m<n$, the representation of $m$ is longer than that of $n$.}
We present two interesting examples out of the seven new $6$-state representations. 
\begin{theorem}
    Let $L = 0^*( \epsilon | 1 | 10 (\epsilon|0|1) 1^* (01^+)^* (\epsilon|0) )$. 
    Then $L$ is complete and unambiguous.
\end{theorem}
\begin{proof}
We use the following {\tt Walnut} code:
\begin{verbatim}
reg one0sq {0,1} "0*(()|1|10(()|0|1)1*(01+)*(()|0))":
eval one0sqTestC "?msd_fib An Ex $one0sq(x) & $fcanon(x,n)":
eval one0sqTestU "?msd_fib ~En,x,y $one0sq(x) & $one0sq(y)
    & (~$equal(x,y)) & $fcanon(x,n) & $fcanon(y,n)":
\end{verbatim}
\noindent Both returned \texttt{TRUE}. 
Here \texttt{one0sq} tests membership in $L$.
\end{proof}
Notice this representation allows $100$ at the very beginning but no other consecutive $0$'s are allowed. 
This restriction on $00$ blocks is very similar to Brown's. 
In fact, Brown's can be expressed, in the form of a regular expression, as\\
\centerline{$0^*( \epsilon | 1 1^* (01^+)^* (\epsilon|0) ) = 0^*( \epsilon | 1 $\st{$| 10 (\epsilon|0|1)$}$ 1^* (01^+)^* (\epsilon|0) )$.}

We can imagine that a new representation could be generated for allowing a block of $00$ after the second $1$, or the third, or after both the first and third $1$, or the first and fourth, etc.
This offers another construction of infinitely many perfect representations.

\begin{theorem}
    Let $L$ be the language accepted by the DFA $Z$. Then $L$ is complete and unambiguous. 
\end{theorem}
\begin{figure}[h]
\centering
    \includegraphics[width=3.5in]{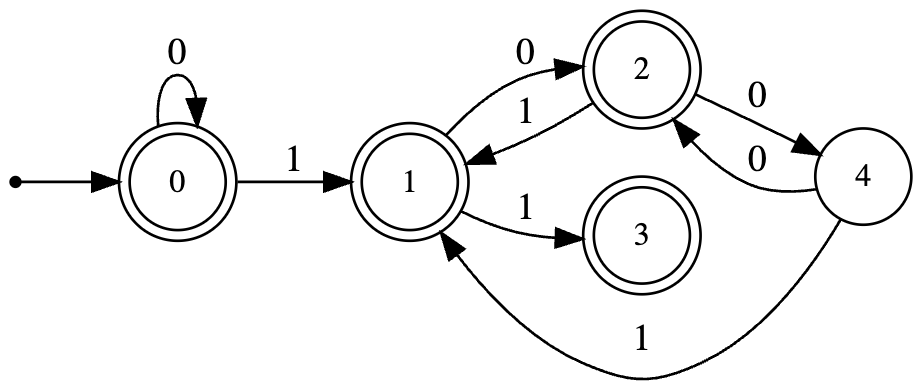}
    \caption{The DFA $Z$.}
    \label{fig: newDFA from search}
\end{figure}
\begin{proof}
We use the following {\tt Walnut} code:
\begin{verbatim}
eval azTestC "?msd_fib An Ex $az(x) & $fcanon(x,n)":
eval azTestU "?msd_fib ~En,x,y $az(x) & $az(y) & (~$equal(x,y))
    & $fcanon(x,n) & $fcanon(y,n)":
\end{verbatim}
\noindent Both returned \texttt{TRUE}. 
Here \texttt{az} tests membership in $L$.
\end{proof}
The strings in $L$ can end with a single $1$ or the block $11$ or an odd number of $0$'s,
but not an even number of $0$'s.
Additionally, the strings cannot contain the block ``$11$'' anywhere but the end.
This restriction on ``$11$'' is reminiscent of the Zeckendorf representation.

\section{Final remarks}
The ideas in this paper can be extended in many different ways.  For example, we could consider representations in terms of Fibonacci numbers of both positive and negative index with various constraints \cite{Park&Cho&Cho&Cho&Park:2020}, or
representations
in terms of sums of the Lucas
numbers \cite{Brown:1969},   or other linear recurrences,
such as the Pell numbers \cite{Horadam:1993}  or Tribonacci numbers \cite{Carlitz&Scoville&Hoggatt:1972}.   The automaton-based
approach can be used in all of these cases.

\newcommand{\noopsort}[1]{} \newcommand{\singleletter}[1]{#1}


\begin{thebibliography}{10}
\providecommand{\bibitemdeclare}[2]{}
\providecommand{\surnamestart}{}
\providecommand{\surnameend}{}
\providecommand{\urlprefix}{Available at }
\providecommand{\url}[1]{\texttt{#1}}
\providecommand{\href}[2]{\texttt{#2}}
\providecommand{\urlalt}[2]{\href{#1}{#2}}
\providecommand{\doi}[1]{doi:\urlalt{https://doi.org/#1}{#1}}
\providecommand{\eprint}[1]{arXiv:\urlalt{https://arxiv.org/abs/#1}{#1}}
\providecommand{\bibinfo}[2]{#2}

\bibitemdeclare{article}{Alpert:2009}
\bibitem{Alpert:2009}
\bibinfo{author}{H.~\surnamestart Alpert\surnameend} (\bibinfo{year}{2009}):
  \emph{\bibinfo{title}{Differences of multiple {Fibonacci} numbers}}.
\newblock {\slshape \bibinfo{journal}{INTEGERS}} \bibinfo{volume}{9},
  \doi{10.1515/INTEG.2009.061}.
\newblock \bibinfo{note}{Paper \#A57}.

\bibitemdeclare{article}{Berstel:2001}
\bibitem{Berstel:2001}
\bibinfo{author}{J.~\surnamestart Berstel\surnameend} (\bibinfo{year}{2001}):
  \emph{\bibinfo{title}{An exercise on {Fibonacci} representations}}.
\newblock {\slshape \bibinfo{journal}{RAIRO Inform. Th\'eor. App.}}
  \bibinfo{volume}{35}, pp. \bibinfo{pages}{491--498},
  \doi{10.1051/ita:2001127}.

\bibitemdeclare{article}{Brown:1965}
\bibitem{Brown:1965}
\bibinfo{author}{J.~L. \surnamestart Brown\surnameend, Jr.}
  (\bibinfo{year}{1965}): \emph{\bibinfo{title}{A new characterization of the
  {Fibonacci} numbers}}.
\newblock {\slshape \bibinfo{journal}{Fibonacci Quart.}}
  \bibinfo{volume}{3}(\bibinfo{number}{1}), pp. \bibinfo{pages}{1--8}.

\bibitemdeclare{article}{Brown:1969}
\bibitem{Brown:1969}
\bibinfo{author}{J.~L. \surnamestart Brown\surnameend, Jr.}
  (\bibinfo{year}{1969}): \emph{\bibinfo{title}{Unique representation of
  integers as sums of distinct {Lucas} numbers}}.
\newblock {\slshape \bibinfo{journal}{Fibonacci Quart.}} \bibinfo{volume}{7},
  pp. \bibinfo{pages}{243--252}.

\bibitemdeclare{article}{Bruyere&Hansel&Michaux&Villemaire:1994}
\bibitem{Bruyere&Hansel&Michaux&Villemaire:1994}
\bibinfo{author}{V.~\surnamestart {Bruy\`ere}\surnameend},
  \bibinfo{author}{G.~\surnamestart Hansel\surnameend},
  \bibinfo{author}{C.~\surnamestart Michaux\surnameend} \&
  \bibinfo{author}{R.~\surnamestart Villemaire\surnameend}
  (\bibinfo{year}{1994}): \emph{\bibinfo{title}{Logic and $p$-recognizable sets
  of integers}}.
\newblock {\slshape \bibinfo{journal}{Bull. Belgian Math. Soc.}}
  \bibinfo{volume}{1}, pp. \bibinfo{pages}{191--238},
  \doi{10.36045/bbms/1103408547}.
\newblock \bibinfo{note}{Corrigendum, {\it Bull.\ Belg.\ Math.\ Soc.} 1 (1994),
  p.~577}.

\bibitemdeclare{article}{Buchi:1960}
\bibitem{Buchi:1960}
\bibinfo{author}{J.~R. \surnamestart {B\"uchi}\surnameend}
  (\bibinfo{year}{1960}): \emph{\bibinfo{title}{Weak second-order arithmetic
  and finite automata}}.
\newblock {\slshape \bibinfo{journal}{Zeitschrift {f\"ur} mathematische Logik
  und Grundlagen der Mathematik}} \bibinfo{volume}{6}, pp.
  \bibinfo{pages}{66--92}, \doi{10.1002/malq.19600060105}.
\newblock \bibinfo{note}{Reprinted in S. Mac Lane and D. Siefkes, eds., {\it
  The Collected Works of J. Richard {B\"uchi}}, Springer-Verlag, 1990, pp.\
  398--424}.

\bibitemdeclare{article}{Bunder:1992}
\bibitem{Bunder:1992}
\bibinfo{author}{M.~W. \surnamestart Bunder\surnameend} (\bibinfo{year}{1992}):
  \emph{\bibinfo{title}{Zeckendorf representations using negative {Fibonacci}
  numbers}}.
\newblock {\slshape \bibinfo{journal}{Fibonacci Quart.}} \bibinfo{volume}{30},
  pp. \bibinfo{pages}{111--115}.

\bibitemdeclare{article}{Carlitz&Scoville&Hoggatt:1972}
\bibitem{Carlitz&Scoville&Hoggatt:1972}
\bibinfo{author}{L.~\surnamestart Carlitz\surnameend},
  \bibinfo{author}{R.~\surnamestart Scoville\surnameend} \&
  \bibinfo{author}{V.E. \surnamestart Hoggatt\surnameend, Jr.}
  (\bibinfo{year}{1972}): \emph{\bibinfo{title}{Fibonacci representations of
  higher order}}.
\newblock {\slshape \bibinfo{journal}{Fibonacci Quart.}} \bibinfo{volume}{10},
  pp. \bibinfo{pages}{43--69, 94}.

\bibitemdeclare{article}{Hajnal:2023}
\bibitem{Hajnal:2023}
\bibinfo{author}{P.~\surnamestart Hajnal\surnameend} (\bibinfo{year}{2023}):
  \emph{\bibinfo{title}{A short note on numeration systems with negative digits
  allowed}}.
\newblock {\slshape \bibinfo{journal}{Bull. Inst. Combin. Appl.}}
  \bibinfo{volume}{97}, pp. \bibinfo{pages}{54--66}.

\bibitemdeclare{book}{Hopcroft&Ullman:1979}
\bibitem{Hopcroft&Ullman:1979}
\bibinfo{author}{J.~E. \surnamestart Hopcroft\surnameend} \&
  \bibinfo{author}{J.~D. \surnamestart Ullman\surnameend}
  (\bibinfo{year}{1979}): \emph{\bibinfo{title}{Introduction to Automata
  Theory, Languages, and Computation}}.
\newblock \bibinfo{publisher}{Addison-Wesley}.

\bibitemdeclare{incollection}{Horadam:1993}
\bibitem{Horadam:1993}
\bibinfo{author}{A.~F. \surnamestart Horadam\surnameend}
  (\bibinfo{year}{1993}): \emph{\bibinfo{title}{Zeckendorf representations of
  positive and negative integers by {Pell} numbers}}.
\newblock In \bibinfo{editor}{G.~E. \surnamestart Bergum\surnameend},
  \bibinfo{editor}{A.~N. \surnamestart Philippou\surnameend} \&
  \bibinfo{editor}{A.~F. \surnamestart Horadam\surnameend}, editors: {\slshape
  \bibinfo{booktitle}{Applications of Fibonacci Numbers}}, \bibinfo{volume}{5},
  \bibinfo{publisher}{Kluwer}, pp. \bibinfo{pages}{305--316},
  \doi{10.1007/978-94-011-2058-6_29}.

\bibitemdeclare{article}{Lekkerkerker:1952}
\bibitem{Lekkerkerker:1952}
\bibinfo{author}{C.~G. \surnamestart Lekkerkerker\surnameend}
  (\bibinfo{year}{1952}): \emph{\bibinfo{title}{Voorstelling van natuurlijke
  getallen door een som van getallen van {Fibonacci}}}.
\newblock {\slshape \bibinfo{journal}{Simon Stevin}} \bibinfo{volume}{29}, pp.
  \bibinfo{pages}{190--195}.

\bibitemdeclare{unpublished}{Mousavi:2016}
\bibitem{Mousavi:2016}
\bibinfo{author}{H.~\surnamestart Mousavi\surnameend} (\bibinfo{year}{2016}):
  \emph{\bibinfo{title}{Automatic theorem proving in {{\tt Walnut}}}}.
\newblock \bibinfo{note}{Arxiv preprint arXiv:1603.06017 [cs.FL], available at
  \url{http://arxiv.org/abs/1603.06017}}.

\bibitemdeclare{article}{Mousavi&Schaeffer&Shallit:2016}
\bibitem{Mousavi&Schaeffer&Shallit:2016}
\bibinfo{author}{H.~\surnamestart Mousavi\surnameend},
  \bibinfo{author}{L.~\surnamestart Schaeffer\surnameend} \&
  \bibinfo{author}{J.~\surnamestart Shallit\surnameend} (\bibinfo{year}{2016}):
  \emph{\bibinfo{title}{Decision Algorithms for {Fibonacci}-Automatic Words,
  {I:} Basic Results}}.
\newblock {\slshape \bibinfo{journal}{RAIRO Inform. Th\'eor. App.}}
  \bibinfo{volume}{50}, pp. \bibinfo{pages}{39--66}, \doi{10.1051/ita/2016010}.

\bibitemdeclare{article}{Ostrowski:1922}
\bibitem{Ostrowski:1922}
\bibinfo{author}{A.~\surnamestart Ostrowski\surnameend} (\bibinfo{year}{1922}):
  \emph{\bibinfo{title}{Bemerkungen zur {Theorie} der {Diophantischen}
  {Approximationen}}}.
\newblock {\slshape \bibinfo{journal}{Abh. Math. Sem. Hamburg}}
  \bibinfo{volume}{1}, pp. \bibinfo{pages}{77--98,250--251},
  \doi{10.1007/BF02940595}.
\newblock \bibinfo{note}{Reprinted in {\it Collected Mathematical Papers},
  Vol.\ 3, pp.\ 57--80}.

\bibitemdeclare{article}{Park&Cho&Cho&Cho&Park:2020}
\bibitem{Park&Cho&Cho&Cho&Park:2020}
\bibinfo{author}{H.~\surnamestart Park\surnameend},
  \bibinfo{author}{B.~\surnamestart Cho\surnameend},
  \bibinfo{author}{D.~\surnamestart Cho\surnameend}, \bibinfo{author}{Y.~D.
  \surnamestart Cho\surnameend} \& \bibinfo{author}{J.~\surnamestart
  Park\surnameend} (\bibinfo{year}{2020}):
  \emph{\bibinfo{title}{Representations of integers as sums of {Fibonacci}
  numbers}}.
\newblock {\slshape \bibinfo{journal}{Symmetry}}
  \bibinfo{volume}{12}(\bibinfo{number}{10}), \doi{10.3390/sym12101625}.
\newblock \bibinfo{note}{Paper 1625}.

\bibitemdeclare{article}{Robbins:1996}
\bibitem{Robbins:1996}
\bibinfo{author}{N.~\surnamestart Robbins\surnameend} (\bibinfo{year}{1996}):
  \emph{\bibinfo{title}{Fibonacci partitions}}.
\newblock {\slshape \bibinfo{journal}{Fibonacci Quart.}} \bibinfo{volume}{34},
  pp. \bibinfo{pages}{306--313}.

\bibitemdeclare{article}{Shallit:2021c}
\bibitem{Shallit:2021c}
\bibinfo{author}{J.~\surnamestart Shallit\surnameend} (\bibinfo{year}{2021}):
  \emph{\bibinfo{title}{Robbins and {Ardila} meet {Berstel}}}.
\newblock {\slshape \bibinfo{journal}{Inform. Process. Lett.}}
  \bibinfo{volume}{167}, \doi{10.1016/j.ipl.2020.106081}.
\newblock \bibinfo{note}{Paper 106081}.

\bibitemdeclare{book}{Shallit:2022}
\bibitem{Shallit:2022}
\bibinfo{author}{J.~\surnamestart Shallit\surnameend} (\bibinfo{year}{2022}):
  \emph{\bibinfo{title}{The Logical Approach to Automatic Sequences: Exploring
  Combinatorics on Words with {\tt Walnut}}}.
\newblock {\slshape \bibinfo{series}{London Math. Soc. Lecture Notes Series}}
  \bibinfo{volume}{482}, \bibinfo{publisher}{Cambridge University Press},
  \doi{10.1017/9781108775267}.

\bibitemdeclare{article}{Zeckendorf:1972}
\bibitem{Zeckendorf:1972}
\bibinfo{author}{E.~\surnamestart Zeckendorf\surnameend}
  (\bibinfo{year}{1972}): \emph{\bibinfo{title}{{Repr\'esentation} des nombres
  naturels par une somme de nombres de {Fibonacci} ou de nombres de {Lucas}}}.
\newblock {\slshape \bibinfo{journal}{Bull. Soc. Roy. {Li\`ege}}}
  \bibinfo{volume}{41}, pp. \bibinfo{pages}{179--182}.

\end{thebibliography}
\end{document}